\documentclass[11pt]{article}
\usepackage{tabularx}
\usepackage{amsmath,amstext,amssymb,amsfonts}
\usepackage{color,graphicx}
\usepackage[usenames,dvipsnames]{xcolor}
\usepackage{enumerate}
 
\usepackage{verbatim}

\usepackage{amsthm}
\usepackage{ifdraft}

\usepackage{nicefrac}

\newtheorem{theorem}{Theorem}[section]

\newtheorem{proposition}[theorem]{Proposition}
\newtheorem{lemma}[theorem]{Lemma}

\theoremstyle{definition}

\newtheorem{algorithm}[theorem]{Algorithm}

\newtheorem{problem}[theorem]{Problem}

\newtheorem{remark}[theorem]{Remark}

\usepackage{boxedminipage}

\newenvironment{mybox}
{\center \noindent\begin{boxedminipage}{1.0\linewidth}}
{\end{boxedminipage}
\noindent
}

\usepackage{prettyref}
\newcommand{\savehyperref}[2]{\texorpdfstring{\hyperref[#1]{#2}}{#2}}

\newrefformat{eq}{\savehyperref{#1}{Equation~{(\ref*{#1})}}}
\newrefformat{lem}{\savehyperref{#1}{Lemma~\ref*{#1}}}
\newrefformat{def}{\savehyperref{#1}{Definition~\ref*{#1}}}
\newrefformat{thm}{\savehyperref{#1}{Theorem~\ref*{#1}}}
\newrefformat{cor}{\savehyperref{#1}{Corollary~\ref*{#1}}}
\newrefformat{cha}{\savehyperref{#1}{Chapter~\ref*{#1}}}
\newrefformat{sec}{\savehyperref{#1}{Section~\ref*{#1}}}
\newrefformat{app}{\savehyperref{#1}{Appendix~\ref*{#1}}}
\newrefformat{tab}{\savehyperref{#1}{Table~\ref*{#1}}}
\newrefformat{fig}{\savehyperref{#1}{Figure~\ref*{#1}}}
\newrefformat{hyp}{\savehyperref{#1}{Hypothesis~\ref*{#1}}}
\newrefformat{alg}{\savehyperref{#1}{Algorithm~\ref*{#1}}}
\newrefformat{sdp}{\savehyperref{#1}{SDP~\ref*{#1}}}
\newrefformat{rem}{\savehyperref{#1}{Remark~\ref*{#1}}}
\newrefformat{item}{\savehyperref{#1}{Item~\ref*{#1}}}
\newrefformat{step}{\savehyperref{#1}{step~\ref*{#1}}}
\newrefformat{conj}{\savehyperref{#1}{Conjecture~\ref*{#1}}}
\newrefformat{fact}{\savehyperref{#1}{Fact~\ref*{#1}}}
\newrefformat{prop}{\savehyperref{#1}{Proposition~\ref*{#1}}}
\newrefformat{claim}{\savehyperref{#1}{Claim~\ref*{#1}}}
\newrefformat{relax}{\savehyperref{#1}{Relaxation~\ref*{#1}}}
\newrefformat{red}{\savehyperref{#1}{Reduction~\ref*{#1}}}
\newrefformat{part}{\savehyperref{#1}{Part~\ref*{#1}}}

\newcommand{\Sref}[1]{\hyperref[#1]{\S\ref*{#1}}}

   \usepackage{fullpage}

\renewcommand{\leq}{\leqslant}

\renewcommand{\geq}{\geqslant}

\usepackage{bm}

\usepackage{xspace}

\usepackage[pdftex,pagebackref,colorlinks,linkcolor=black,filecolor = black, citecolor = black, urlcolor  = blue]{hyperref}

\newcommand{\mper}{\,.}

\newcommand{\cG}{\mathcal G}

\newcommand{\cI}{\mathcal I}

\newcommand{\cN}{\mathcal N}

\newcommand{\Brac}[1]{\left[#1 \right]}

\newcommand{\set}[1]{\{#1\}}

\newcommand{\Abs}[1]{\left\lvert#1\right\rvert}

\newcommand{\ceil}[1]{\lceil #1 \rceil}
\newcommand{\floor}[1]{\lfloor #1 \rfloor}

\newcommand{\norm}[1]{\lVert#1\rVert}

\newcommand{\defeq}{\stackrel{\textup{def}}{=}}

\newcommand{\inprod}[1]{\langle #1\rangle}

\newcommand{\Z}{{\mathbb Z}}

\newcommand{\R}{\mathbb R}

\newcommand{\Esymb}{\mathbb{E}}
\newcommand{\Psymb}{\mathbb{P}}

\DeclareMathOperator*{\E}{\Esymb}

\DeclareMathOperator*{\ProbOp}{\Psymb}

\newcommand{\given}{\mathrel{}\middle|\mathrel{}}
\newcommand{\Given}{\given}

\newcommand{\Ex}[1]{\E\Brac{#1}}

\renewcommand{\Pr}[1]{\ProbOp\Brac{#1}}

\newcommand{\e}{\epsilon}

\definecolor{DSgray}{cmyk}{0,0,0,0.7}

\let\e\varepsilon

\newcommand{\bigO}{\mathcal{O}}
\newcommand{\U}{ \bar{u}}
\newcommand{\V}{ \bar{v}}

\newcommand{\phik}{{\phi^k_G}}

\newcommand{\skp}{{\sf sparsest}\ $k$-{\sf partition} }

\title{Approximation Algorithm for Sparsest $k$-Partitioning}
\author{Anand Louis 
\thanks{Supported in part by NSF awards CCF-0915903 and CCF-1217793. }
\\ Georgia Tech \\ anandl@gatech.edu
 \and 
Konstantin Makarychev \\ Microsoft Research \\ komakary@microsoft.com }

\date{}
\begin{document}
\begin{titlepage}

\maketitle

\begin{abstract}
Given a graph $G$, the sparsest-cut problem asks to find the set of vertices $S$ which
has the least expansion defined as
$$\phi_G(S) \defeq \frac{w(E(S,\bar{S}))}{ \min \set{w(S), w(\bar{S})}  } , $$
where $w$ is the total edge weight of a subset.
Here we study the natural generalization of this problem: given an integer $k$,
compute a $k$-partition $\set{P_1, \ldots, P_k}$ of the vertex set so as to minimize 
$$  \phik(\set{P_1, \ldots, P_k}) \defeq \max_i \phi_G(P_i). $$
Our main result is a polynomial time bi-criteria approximation algorithm which outputs
a $(1 - \e)k$-partition of the vertex set such that each piece has expansion at 
most $\bigO_{\e}(\sqrt{\log n \log k})$ times $OPT$. 
We also study balanced versions of this problem.

\end{abstract}

\end{titlepage}

\section{Introduction}

The Sparsest Cut problem asks to find a subset $S$ of vertices of a given graph $G=(V,E)$ such that
the total weight of edges leaving it is as small as possible compared to its size. This latter quantity, called
expansion or conductance, is defined as:
$$ \phi_G(S) = \frac{|E(S,\bar{S})|}{w(S)},$$
where $|E(S,\bar{S})|$ is the size of the cut $E(S,\bar{S})$ and $w(S)=\sum_{u\in S} w_s$ is the weight of vertices in $S$.
Typically, the weight of each vertex $u$ is its degree $d_u$. The optimal value is called the expansion of the graph $G$, and
is denoted by $\phi_G$:
$$ \phi_G = \min_{S : w(S) \leq w(V)/2 } \phi_G(S).$$

The Sparsest Cut problem has been highly influential
in the study of algorithms and complexity. 
In their seminal work, Leighton and Rao~\cite{lr99} gave a $\bigO(\log n)$ approximation algorithm
for the problem. Later, Linial, London, Rabinovich~\cite{llr95} and Aumann, Rabani~\cite{ar98} gave $\bigO(\log n)$
approximation
for Sparsest Cut with non-uniform demands and established a connection between the Sparsest Cut problem
and embeddings of metric spaces into $\ell_2$.
In a breakthrough result, Arora, Rao, and Vazirani~\cite{arv04}
gave $\bigO(\sqrt{\log n})$ approximation algorithm for the problem. Given a graph with optimal sparsest cut $OPT=\phi_G$, 
their algorithm returns a set $S_{A}$ with $\phi_G(S_{A}) = \bigO(\sqrt{\log n}\, OPT)$. 

The fundamental Cheeger's inequality (shown for graphs in \cite{a86, am85}) establishes a
bound on expansion via the spectrum of the graph. In particular, it states that a natural 
spectral algorithm gives a solution of cost $\phi_G(S_{A}) = \bigO(\sqrt{OPT})$. (Note that $OPT \leq 1$, so 
$\sqrt{OPT} \geq OPT$.)

Many extensions of this problem have been considered in the literature
(see \prettyref{sec:related} for a brief survey).
In this work, we study a very natural extension of the Sparsest Cut problem -- the Sparsest $k$-Partitioning problem.

\begin{problem}[Sparsest $k$-Partitioning Problem]
Given a graph $G = (V,E)$ and a parameter $k$, compute a partition $\set{P_1, \ldots, P_k}$ of $V$  
into $k$ non-empty pieces so as to minimize 
\[ \phik(\set{P_1, \ldots, P_k}) =  \max_i \phi_G(P_i). \]
The optimal value $OPT$ is called the $k$-sparsity and denoted by $\phi^k_G$.
\end{problem}

The problem gained prominence because of its close connection to the graph spectrum. 
This connection was established and studied in the recent works~\cite{lot12, lrtv12}
which motivate the study of $\phik$ as a combinatorial analogue of $\lambda_k$
(the $k$-th smallest eigenvalue of the normalized Laplacian of the graph $G$).
Lee, Gharan and Trevisan~\cite{lot12} showed that there exists a $k$-partition 
$\set{P_1, \ldots, P_k}$ such that $\phik(\set{P_1, \ldots, P_k}) \leq \bigO( k^3 \sqrt{\lambda_k})$.
Louis, Raghavendra, Tetali, Vempala~\cite{lrtv12}
showed that for any $k$ non-empty disjoint subset $S_1, \ldots, S_k$ of $V$, we have $\max_i \phi(S_i) \geq \Omega(\lambda_k)$.
Moreover, they showed that for some absolute constant $c$, 
there exists $k$ disjoint non-empty sets $S_1, \ldots S_{k} \subset V(G)$ such that 
$ \max_i \phi(S_i) \leq \bigO(\sqrt{ \lambda_{ck} \log k })$. Lee et al.~\cite{lot12} proved a similar result 
with $c = 1 +\e$ for any $\e > 0$. 
Note that in both these results, the sets $\set{S_i : i \in [k]}$
need not form a partition of the vertex set of the graph. 
As a by-product of our main result, we 
slightly strengthen the results above. 

\begin{proposition}
\label{prop:cheeger-partition}
Given a graph $G$ and a parameter $k$,
\[ \lambda_k \leq \phik \leq \bigO_{\e}\left(\sqrt{ \lambda_{(1+\e)k} \log k }\right)\]
for every $\e>0$. Here and below, $\bigO_{\e}(f)$ denotes $\bigO({\sf poly}(1/\e) f)$. 
\end{proposition}

No approximation algorithm for $\phik$ with a multiplicative guarantee was known prior to our work.
In this paper, we prove the following theorems.

\begin{theorem}
\label{thm:kpart-alg-1}
There exists a probabilistic polynomial-time algorithm that 
given an undirected graph $G=(V,E)$ with arbitrary vertex weights $w_u$ and parameters $k\in \Z^+$ $(k\geq 2)$, $\e>0$,
outputs $k'\geq (1 - \e)k$ partition such that each set has expansion at most
$\bigO_{\e}
\left(\sqrt{\log n \log k} \; OPT \right)$ w.h.p. Here $OPT$ is the cost of the optimal solution for the Sparsest $k$-Partitioning
problem.
\end{theorem}

\begin{theorem}
\label{thm:kpart-alg-2}
There exists a probabilistic polynomial-time algorithm that 
given an undirected graph $G=(V,E)$ 
 with weights $w_u=d_u$ ($d_u$ is the degree of the vertex $u$) 
and parameters $k\in \mathbb{N}$ $(k\geq 2)$, $\e>0$,
outputs $k'\geq (1 - \e)k$ partition such that each set has expansion at most 
$\bigO_{\e} \left(\sqrt{OPT \log k} \right)$ w.h.p.
\end{theorem}

\begin{remark}
Both theorems can be easily extended to edge-weighted graphs. 
W.l.o.g, we may assume that the weights of the edges are integers. 
The proofs of the theorems simply follow by introducing parallel edges: if $w_e \in \Z$ denotes the weight of an edge $e$,
we replace $e$ with $w_e$ unweighted parallel edges. 
No changes are needed in the algorithm, and the algorithm still runs in polynomial time.
\end{remark}

Note that for $k=2$, \prettyref{thm:kpart-alg-1} gives the same guarantee as that of 
Arora, Rao and Vazirani~\cite{arv04} for 
Sparsest Cut and \prettyref{thm:kpart-alg-2} gives the same guarantee as that of Cheeger's inequality \cite{am85,a86}
for Sparsest Cut. A direct corollary of the work of Raghavendra, Steurer and Tulsiani \cite{rst12} is that \prettyref{thm:kpart-alg-2}
is optimal under the SSE hypothesis. We refer the reader to \cite{rs10,rst12} for the statement and implications of the
SSE hypothesis.

\paragraph{SDP Relaxation.}
The proofs of our main theorems go via an SDP relaxation of $\phik$ and a rounding algorithm for it.
As a first attempt, one would try an assignment SDP \`{a} la Unique Games (as used in \cite{k02,t08,cmm06a,cmm06b}),
but such relaxations have a large integrality gap (see \prettyref{app:ug-ig}). 
The main difficulty in constructing an integer programming formulation of \skp is that 
we do not know the sizes of the sets in the optimal partition.
We use a novel SDP relaxation which gets around this obstacle.
In this SDP, we manage to encode a partitioning of the 
graph as well as a special measure on the vertices. This measure tells us how large every set must be. Roughly 
speaking, we expect that in the solution obtained by the algorithm, the measure of
every set is approximately $1$, irrespective of its size. We give a formal description of the SDP in \prettyref{sec:SDP}.

\subsection{Related Work}
\label{sec:related}
The Small Set Expansion problem (SSE) asks to find a set $S$ of weight at 
most $w(S)/k$ (where $k\geq 2$ is a parameter) with the smallest expansion $\phi_G(S)$. This problem got a 
lot of attention recently, partially because of the observed connection with the Unique Games Conjecture~\cite{rs10,abs10}
and because of a new Small Set Expansion Conjecture of Raghavendra and Steurer~\cite{rs10}. 
Raghavendra, Steurer, and Tetali~\cite{rst10}
gave an algorithm with a Cheeger--type approximation guarantee of 
$ \bigO(\sqrt{ OPT \log k})$ for this problem.
Bansal et al.~\cite{bfk11} gave a $\bigO(\sqrt{\log n \log k})$ approximation algorithm 
(i.e., $\phi_G(S_A) \leq \bigO(\sqrt{\log n \log k} \, OPT)$). 
The min-sum version of graph multi-partitioning has also been studied extensively, see e.g. \cite{kvv04,ar06,kns09,lrtv11}.

\paragraph{Comparison to Previous Work}
Bansal et al.~\cite{bfk11} studied the problem of partitioning the graph in $k$ equal pieces
while minimizing the largest edge boundary of the piece  (Min Max Graph Partitioning). They give a bi-criteria 
approximation algorithm 
where each set in the partition is of size at most $2n/k$ while approximating its 
edge boundary to within a $\bigO(\sqrt{\log n \log k})$ factor of $OPT$. 
This problem is somewhat related to ours. However, the crucial difference is that the optimal
solution to our problem may contain sets of very different sizes: large and small. This makes 
their algorithm and SDP relaxation non applicable in our settings.
Since the aim of Min Max Graph Partitioning is to make all edge boundaries small,
the algorithm of Bansal et al.~\cite{bfk11} 
may sometimes add very small sets $P_i$ to the partition being output. Such sets can have large expansion 
inspite of having a small edge boundary. 
The main challenge in Min-Max Graph Partitioning
is to find sets that (a) are of size at most $n/k$; and (b) cover all vertices (without these conditions,
Min Max Graph Partitioning admits a simple constant factor approximation).
In some sense, we need to ensure that the sets 
are {\em not too small} (rather than not too large), and hence the expansion is small. 
As Bansal et al., we also need to cover all vertices, but this is a relatively easy task in our case. In fact,
we first drop this condition altogether and find a collection of disjoint non-expanding sets; then 
we transform these sets into a partitioning.
We also note that our algorithm solves the SDP relaxation only once. 
This is again in contrast with ~\cite{bfk11}, where the SDP relaxation is actually just a relaxation for the SSE problem 
(and not for Min Max Graph Partitioning!). So the Min Max Graph Partitioning algorithm has to solve the SDP relaxation 
at least once for each set in the partitioning (in fact, $\bigO(\log n)$ times). 

As we note above a natural assignment SDP relaxation has a large integrality gap (see \prettyref{app:ug-ig}). To round our new SDP (see 
\prettyref{sec:SDP}), 
one can try to adopt the rounding algorithms of Lee et al.~\cite{lot12} and Louis et al.~\cite{lrtv12}
\footnote{Both \cite{lrtv12,lot12}  construct an embedding of the graph into $\R^k$ as a first step. 
The proofs of their main theorems can be viewed as an algorithm to round these vectors into sets. }.
However, these algorithms could only possibly give an approximation guarantee of the form $\bigO(\sqrt{OPT \log k})$. 
To get rid of the square root, we need to embed the SDP solution from $\ell^2_2$ to $\ell_2$. This step distorts the vectors, so that they no longer satisfy 
SDP constraints and no longer have properties required by these algorithms.

\subsection{Extensions}

Our SDP formulation and rounding algorithm can be used to solve other problems 
as well.
Consider the balanced version of Sparsest $k$-Partition.
\begin{problem}[Balanced Sparsest $k$-Partitioning Problem]
Given a graph $G = (V,E)$ and a parameter $k$, compute a partition $\set{P_1, \ldots, P_k}$ of $V$  
into $k$ non-empty pieces each of weight $w(G)/k$ so as to minimize 
\[ \phik(\set{P_1, \ldots, P_k}) =  \max_i \phi_G(P_i). \]
\end{problem}
Using our techniques, we can prove the following theorems. 
\begin{theorem}
\label{thm:bkpart-alg-1}
There exists a probabilistic polynomial-time algorithm that 
given an undirected graph $G=(V,E)$ with arbitrary vertex weights $w_u$ and parameters $k\in \mathbb{N}$ $(k\geq 2)$, $\e>0$,
outputs $k'\geq (1 - \e)k$ disjoint sets (not necessarily a partition) such that the weight of each set is in the range 
$[w(G)/(2k), (1+\e)w(G)/k]$,
and the expansion of each set is at most 
$\bigO_{\e} \left(\sqrt{\log n \log k} \; OPT \right)$ w.h.p.
\end{theorem}

\begin{theorem}
\label{thm:bkpart-alg-2}
There exists a probabilistic polynomial-time algorithm that 
given an undirected graph $G=(V,E)$ 
with weights $w_u=d_u$ ($d_u$ is the degree of the vertex $u$) 
and parameters $k\in \mathbb{N}$ $(k\geq 2)$, $\e>0$,
outputs $k'\geq (1 - \e)k$ disjoint sets (not necessarily a partition) 
such that the weight of each set is in the range $[w(G)/(2k), (1+\e)w(G)/k]$,
and the expansion of each set is at most 
$\bigO_{\e} \left(\sqrt{OPT \log k} \right)$ w.h.p.
\end{theorem}

Note that the algorithms above return $k'$ disjoint sets that do not have to cover all vertices. The proofs of these 
theorems are similar to the proofs of our main results -- \prettyref{thm:kpart-alg-1} and \prettyref{thm:kpart-alg-2}. We refer the reader to 
\prettyref{sec:merging-sets} for more details. In fact, the assumption that all sets in the optimal solution have the same size makes the balanced problem much simpler. \prettyref{thm:bkpart-alg-1} also follows (possibly with slightly worse guarantees) from the result of 
Krauthgamer, Naor, and Schwartz~\cite{kns09}, who gave a bi-criteria $O(\sqrt{\log n \log k})$
approximation algorithm for the $k$-Balanced Partitioning Problem (with the ``min-sum'' objective).

\subsection{Organization}

We prove \prettyref{thm:kpart-alg-1} in \prettyref{sec:main-alg}. 
We present the SDP relaxation of \skp in \prettyref{sec:SDP} and the main rounding algorithm in 
\prettyref{sec:main-alg}.
We prove \prettyref{thm:kpart-alg-2} in \prettyref{app:kpart-alg-2}.
We prove \prettyref{prop:cheeger-partition} in \prettyref{app:cheeger-partition}.



\section{Main Algorithm}
\label{sec:main=alg}
We first prove a slightly weaker result. We give an algorithm that finds at least $(1-\e)k$ disjoint 
sets each with expansion at most $\bigO_{\e}\left(\sqrt{\log n \log k}\; OPT \right)$. Note that 
we do not require that these sets cover all vertices in $V$.

\begin{theorem}
\label{thm:k-weak-part-alg-1}
There exists a probabilistic polynomial-time algorithm that given an undirected graph $G$ and parameters 
$k\in \mathbb{N}$ $(k\geq 2)$, $\e>0$, outputs $k'\geq (1 - \e)k$ disjoint sets $P_1,\dots, P_{k'}$ such that 
$$\Ex{\max_{i} \phi (S_i)}\leq \bigO_{\e} \left(\sqrt{\log n \log k} \; OPT \right),$$
where $OPT$ is the cost of the optimal sparsest $k$-partitioning of $G$.
\end{theorem}

Then, in \prettyref{sec:partit}, we show how using $k'\geq (1-\e)k$ such sets, we can find a partitioning 
of $V$ into $k''\geq (1-2\e)k$ sets with each set having expansion at most 
$\bigO_{\e}\left(\sqrt{\log n \log k} \;OPT \right)$.

Our algorithm works in several phases. First, it solves the SDP relaxation, which we present in \prettyref{sec:SDP}.
Then it transforms all vectors to unit vectors and  defines a measure $\mu(\cdot)$ on vertices of the graph.
We give the details of this transformation in \prettyref{sec:norm}. Succeeding this, in the main phase, 
the algorithm samples many independent {\em orthogonal separators} $S_1,\dots, S_T$ and then extracts
$k' > (1 - \e)k$ disjoint subsets from them. We describe this phase in \prettyref{sec:main-alg}.
Finally, the algorithm merges some of these sets with the left over vertices to obtain a $k'' \geq (1 - \e)k'$ partition. 
We describe this phase in \prettyref{sec:merging-sets}.

\subsection{SDP Relaxation}
\label{sec:SDP}
We employ a novel SDP relaxation for the \skp problem. The main challenge in writing an SDP relaxation is that we do not know the sizes 
of the sets in advance, so we cannot write standard spreading constraints or spreading constraints used in the paper of 
Bansal et al. \cite{bfk11}. For each vertex $u$,
we introduce a vector $\U$. In the integral solution corresponding to the optimal partitioning $P_1, \dots, P_k$,
each vector $\U$ has $k$ coordinates, one for every set $P_i$:
\[ \U(i) = \begin{cases} \frac{1}{\sqrt{w(P_i)}} & \text{if } u\in P_i;\\0&\text{otherwise.}\end{cases} \]
Observe, that the integral solution satisfies two crucial properties: for each set $P_i$,
\begin{equation}\label{eq:sum-one}
\sum_{u\in P_i} w_u \|\bar  u\|^2 = \sum_{u\in P_i} \frac{w_u}{w(P_i)} = 1,
\end{equation}
and for every vertex $u\in P_i$,
\begin{equation}\label{eq:sparse-new}
\sum_{v\in V} w_v \inprod{ \U, \V } = \sum_{v\in P_i} \frac{w_v}{w(P_i)} + \sum_{v\notin P_i} 0 = 1.
\end{equation}
\prettyref{eq:sum-one} gives us a way to measure sets. Given a set of vectors $\{\U\}$, we define a measure
$\mu(\cdot)$ on vertices as follows
\begin{equation}\label{eq:mu}
\mu(S) = \sum_{u\in S} w_u \|\bar  u\|^2.
\end{equation}
For the intended solution, we have $\mu(P_i)=1$, and hence $\mu(V)=k$. This is the first constraint we add to the SDP:
$$\mu(V)\equiv \sum_{u\in V} w_u \|\bar  u\|^2 = k.$$
From~\prettyref{eq:sparse-new}, we get a spreading constraint:
$$\sum_{v\in V} w_v \inprod{ \U, \V } = 1.$$
We also add $\ell_2^2$ triangle inequalities to the SDP. It is easy to check that they are satisfied
in the intended solution (since they are satisfied for each coordinate).

Finally, we need to write the objective function that measures the expansion of the sets. In the intended solution, 
if $u,v\in P_i$ (for some $i$), then $\U = \V$, and $\|\U - \V\|^2=0$. If 
$u\in P_i$ and $v\in P_j$ (for $i\neq j$), then 
\[ \|\U - \V\|^2 = \|\U\|^2 + \|\V\|^2 = 1/w(P_i) + 1/w(P_j) \mper \] 
Hence,
\[
 \frac{1}{k}\sum_{(u, v)\in E} \|\U - \V\|^2 = \frac{1}{k}\sum_{i < j} \sum_{\substack{(u, v)\in E\\u\in P_i\\v\in P_j}} 
\Big(\frac{1}{w(P_i)} + \frac{1}{w(P_j)}\Big) 
 = \frac{1}{k}\sum_{i}  \frac{|E(P_i, V\setminus P_i)|}{w(P_i)} = \frac{1}{k}\sum_{i} \phi_G(P_i)\leq OPT \mper \]

We get the following SDP relaxation for the problem.
\vskip 1em

\begin{figure}[ht]
\begin{tabularx}{\columnwidth}{|X|}
\hline
\renewcommand{\arraystretch}{1.5}
$$\min \frac{1}{k}\sum_{(u, v)\in E} \|\U - \V\|^2$$
$$\begin{array}{rcll}
\displaystyle{\sum_{u \in V} w_u \|\U\|^2} & = & \displaystyle{k} &\\ 
\displaystyle{\sum_{v \in V} w_v \inprod{ \U, \V}} &= &\displaystyle{1} &\quad \forall u \in V\\
\displaystyle{\|\U - \bar x\|^2 + \|\bar x - \V\|^2} &\geq& \displaystyle{\|\U - \V\|^2}&\quad \forall u,v,x \in V\\
\displaystyle{0 \quad\leq \quad\inprod{\U,\V}} & \leq & \displaystyle{\|\U\|^2} &\quad \forall u,v \in V
\end{array}$$
\\
\hline 
\end{tabularx}
\caption{SDP Relaxation for Sparsest $k$-Partition}
\label{sdp:kpart-sdp}
\end{figure}

\subsection{Normalization}\label{sec:norm}
After the algorithm solves the \prettyref{sdp:kpart-sdp}, we define the measure $\mu$ using \prettyref{eq:mu}, and ``normalize'' all vectors using 
a transformation $\psi$ from the paper of Chlamtac, Makarychev and Makarychev~\cite{cmm06b}. The transformation $\psi$ defines the 
inner products between $\psi(\U)$ and $\psi(\V)$ as follows 
(all vectors $\U$ are nonzero in our SDP relaxation):
$$\inprod{\psi(\U), \psi(\V)}  = \frac{\inprod{ \U,\V }}{ \max \{ \norm{\U}^2 , \norm{\V}^2 \}}.$$
This uniquely defines vectors $\psi(\U)$ (up to an isometry of $\ell_2$). 
Chlamtac, Makarychev and Makarychev showed that the image $\psi(X)$ of any $\ell_2^2$ space $X$ is an $\ell_2^2$ space, and the following condtions hold.
\begin{itemize}
\item For all non-zero vectors $\U \in X$, $\norm{\psi(\U)}^2 = 1 $.
\item For all non-zero vectors $u,v \in X$,
$$\norm{ \psi(\U)  - \psi(\V)   }^2 \leq  \frac{2 \norm{\U - \V}^2 }{ \max \set{ \norm{\U}^2, \norm{\V}^2 }  }.$$
\end{itemize}

\subsection{Orthogonal Separators}

Our algorithm uses the notion of orthogonal separators introduced by Chlamtac, Makarychev, and Makarychev~\cite{cmm06b}. 
Let $X$ be an $\ell_2^2$ space. We say that a distribution over subsets of $X$ is a $k$-orthogonal
separator of $X$ with distortion $D$, probability scale $\alpha >0$ and separation threshold $\beta<1$, if the following conditions hold for $S \subset X$
chosen according to this distribution:
\begin{enumerate}
\item For all $\U \in X$, $\Pr{ \U \in S} = \alpha \norm{\U}^2$.

\item For all $\U, \V \in X$ with $\inprod{\U,\V} \leq \beta \max \set{ \norm{\U}^2, \norm{\V}^2 }$,

$$\Pr{\U \in S \textrm{ and } \V \in S} \leq \frac{ \alpha \min \set{\norm{\U}^2, \norm{\V}^2   }   }{k} .$$

\item For all $u,v \in X$
$$\Pr{ I_S(\U) \neq I_S(\V) } \leq \alpha D \norm{\U - \V}^2. $$
Here $I_S$ is the indicator function
\footnote{I.e., $I_S(\U) \defeq \begin{cases}  1 & \textrm{ if } \U \in S \\ 0; & \textrm{ otherwise.}    \end{cases}  $} of the set $S$. 
\end{enumerate}

\begin{theorem}[\cite{cmm06b,bfk11}]
\label{thm:orthsep}
There exists a polynomial-time randomized algorithm that given a set
of vectors $X$, a parameter $k$, and $\beta < 1$ generates a $k$-orthogonal separator with distortion
$D = \bigO_{\beta} \left( \sqrt{\log \Abs{X} \log k}   \right)$
and scale $\alpha \geq 1/p(|X|)$ for some polynomial $p$.
\end{theorem}

In the algorithm, we sample orthogonal separators from the set of normalized vectors 
$\set{\psi(\U): u\in V}$. For simplicity of exposition we assume that 
an orthogonal separator $S$ contains not vectors $\U$, but
the corresponding vertices. That is, for an orthogonal separator $\tilde{S}$, 
we consider the set of vertices $S = \{u\in V: \psi(\U) \in \tilde{S}\}$.

\subsection{Algorithm}
\label{sec:main-alg}

We give an algorithm for generating $k'\geq (1-\e)k$ disjoint sets $P_i$ in Figure~\ref{alg:main-alg}.

\begin{figure*}[t]
\centering

\begin{mybox}
\begin{algorithm}~
\begin{enumerate}

\item Solve \prettyref{sdp:kpart-sdp} and obtain vectors $\{\U\}$.

\item Compute normalized vectors $\psi(\U)$, and define the measure $\mu(\cdot)$ (see \prettyref{sec:norm} and~Eq.~\prettyref{eq:mu}).

\item Sample $T = 2n/\alpha$ independent $(12k/\e)$-orthogonal separators $S_1, \ldots, S_{T}$ 
for vectors $\psi(\U)$ ($u\in V$) \\ with
separation threshold $\beta = 1 - \e/4$.

\item For each $i$, define $S_i'$ as follows:
$$ S_i' = \begin{cases} S_i & \textrm{if }  \mu(S_i) \leq 1 + \e/2; \\ \varnothing & \textrm{otherwise}.  \end{cases} $$

\item For each $i$, let $S_i'' = S'_i \setminus \left( \cup_t^{i-1} S_t'  \right)$ be the set of yet uncovered vertices in $S'_i$.

\item For each $i$, set $P_i =\{u\in S''_i: \|\U\|^2\geq r_i\}$, where the parameter $r_i$ is chosen to minimize
the \\ expansion $\phi_G(P_i)$ of the set $P_i$.

\item Output $(1-\e)k$ non-empty sets $P_i$ with the smallest expansion $\phi_G(P_i)$.
\end{enumerate}
\label{alg:ksets}
\end{algorithm}
\end{mybox}
\caption{Algorithm for generating $k'\geq (1-\e)k$ disjoint sets $P_i$.}
\label{alg:main-alg}
\end{figure*}

\subsection{\texorpdfstring{Properties of Sets $S''_i$}{Properties of Sets S}}
We prove that (a) the edge boundaries of the sets $S_i''$ are {\em small}; and 
(b) the sets $S''_i$ form a partition of $V$ w.h.p. The following lemma makes these statements precise.

\begin{lemma}\label{lem:SP-prop} For a set $S\subset V$, define 
\begin{equation}\label{def:nu}
\nu (S) = \sum_{\substack{(u,v)\in E(S,V\setminus S) \\ u \in S, v \notin S}} \|\U\|^2 + 
	\sum_{ \substack{  (u,v)\in E \\ u,v \in S}} |\|\U\|^2 - \|\V\|^2|.
\end{equation}
Then, sets $S_i''$ satisfy the following conditions:
\begin {enumerate}[(a)]
\item \[ \Ex{\sum_i \nu(S''_i)} \leq (8D +1)k \cdot SDP,\]
 where $D=\bigO_{\e}(\sqrt{\log n \log k})$ is the distortion of $(12k/\e)$-orthogonal separator, and $SDP$ is the value
of the SDP solution.
\item All sets $ S''_i$ are disjoint; and 
	\[ \Pr{\mu(\cup S''_i) = k}\geq 1-ne^{-n} \mper \]
\end{enumerate}
\end{lemma}
\begin{proof}
(a) Let $E_{cut}$ be the set of edges cut by the partitioning $S_1'',\dots, S_T'', V\setminus (\cup S_i'')$. Observe, that each
cut edge $(u,v)$ contributes $\norm{\U}^2 + \norm{\V}^2$ to the sum $\sum \nu(S_i'')$, and each uncut 
edge contributes either $|\|\U\|^2 - \|\V\|^2|$, or $0$. Hence,
\[  \Ex{\sum_i \nu(S_i'')}  \leq  \Ex{\sum_{(u,v)\in E_{cut}} (\|\U\|^2 + \|\V\|^2)}  
  + \sum_{(u,v)\in E} |\|\U\|^2 - \|\V\|^2|  \]
The second term is bounded by 
$$\sum_{(u,v)\in E}\|\U - \V\|^2 = k\cdot SDP,$$ since 
$$\|\U\|^2 - \|\V\|^2= \|\U - \V\|^2 - 2(\|\V\|^2 - \inprod{ \U, \V}) 
\leq \|\U - \V\|^2.$$
The inequality follows from the SDP constraint $\|\V\|^2 \geq \inprod{ \U, \V}$.
We now bound the first term. To do so we need the following lemma.
\begin{lemma}\label{lem:min-prob}
For every vertex $u\in V$ and $i\in\{1,\dots, T\}$, we have
$\Pr{u\in S_i'}\geq \alpha /2$.
\end{lemma}

We give the proof of \prettyref{lem:min-prob} after we finish the proof of \prettyref{lem:SP-prop}. 
Let us estimate the probability that an edge $(u,v)$ is cut. Let $U_t = \cup_{i\leq t}S'_i$ be the set of vertices covered by the first $t$ sets $S'_i$. 
Note, that $S''_i = S'_i\setminus U_{i-1}$.
We say that the edge $(u,v)$ is cut by the set $S'_t$, if $S'_t$ is the first set containing $u$ or $v$, and it contains only one of these vertices.
Then,
\begin{eqnarray*}
 \Pr {(u,v)\in E_{cut}} &= & \sum_{i=1}^T \Pr {(u,v)\text{ is cut by }S'_i\,} \\ 
 &= & \sum_{i=1}^T \Pr {u,v \notin U_{i-1} \text{ and } I_{S'_i}(u) \neq I_{S'_i}(v)}  \\
& \leq & \sum_{i=1}^T \Pr {u \notin U_{i-1} \text{ and } I_{S_i}(u) \neq I_{S_i}(v)} \\
 & = & \sum_{i=1}^T \Pr {u \notin U_{i-1}} \Pr{I_{S_i}(u) \neq I_{S_i}(v)}. 
\end{eqnarray*}
Now, by \prettyref{lem:min-prob}, $\Pr {u \notin U_{i-1}}\leq (1-\alpha/2)^{i-1}$, and, by Property 3 of orthogonal separators,
\begin{eqnarray*}
\Pr{I_{S_i}(u) \neq I_{S_i}(v)} &\leq & \alpha D\|\psi(u)-\psi(v)\|^2 \\
  &\leq & \frac{2\alpha D\|\U-\V\|^2}{\max\{\|\U\|^2, \|\V\|^2\}}.
\end{eqnarray*}
Thus  (using $\sum_i(1-\alpha/2)^i \leq 2/\alpha$),
$$\Pr {(u,v)\in E_{cut}} \leq \frac{4D\,\|\U-\V\|^2}{\max\{\|\U\|^2, \|\V\|^2\}}.$$
We are almost done,
\begin{eqnarray*}
 \Ex{\sum_{(u,v)\in E_{cut}} (\|\U\|^2 + \|\V\|^2)}
& = & \sum_{(u,v)\in E} \Pr {(u,v)\in E_{cut}} (\|\U\|^2 + \|\V\|^2) \\
& \leq & \sum_{(u,v)\in E} \frac{4D\,\|\U-\V\|^2}{\max\{\|\U\|^2, \|\V\|^2\}}\cdot 
(\|\U\|^2 + \|\V\|^2)  \\
&\leq & \sum_{(u,v)\in E} 8D\,\|\U-\V\|^2 = 8kD \cdot SDP.  
\end{eqnarray*}
Thus we get that 
\[ \Ex{\sum_i \nu(S''_i)}\leq (8D + 1)k\cdot SDP \mper \] 

\medskip

(b) The sets $S_i''$ are disjoint by definition. By \prettyref{lem:min-prob}, the probability that a vertex is not covered by any set $S_i$ is 
$(1-\alpha/2)^T = (1-\alpha/2)^{2n/\alpha} < e^{-n}$. So with probability at least $1- n e^{-n}$ all vertices are covered.
\end{proof}

It remains to prove \prettyref{lem:min-prob}.
\begin{proof}[Proof of \prettyref{lem:min-prob}] We adopt a slightly modified argument from the paper of 
Bansal et al. \cite{bfk11} (Theorem~2.1, arXiv). If $u\in S_i$, then $u\in S_i'$ unless $\mu(S_i) > 1 + \e/2$, hence
\begin{align*}
\Pr{u\in S'_i\,} & =  \Pr{u\in S_i} (1 - \Pr{\mu(S_i)> 1+ \e/2 \Given u\in S_i}) \\
 & =  \alpha (1 - \Pr{\mu(S_i)> 1+{\e/2} \Given u\in S_i}).
\end{align*}
Here, we used that $\Pr{u\in S_i} = \alpha \|\psi(\U)\|^2 = \alpha$ (see Property 1 of orthogonal separators).
We need to show that $\Pr{\mu(S_i)> 1+\e/2 \Given u\in S_i}\leq 1/2$.
Let us define the sets $A_u$ and $B_u$ as follows. 
$$ A_u = \set{ v \in V: \inprod{ \psi(\U), \psi(\V)} \geq \beta  } $$
and 
$$ B_u = \set{v \in V: \inprod{ \psi(\U), \psi(\V)}  < \beta }.  $$
Now, 
\begin{align*}
 \mu(A_u)  &=   \sum_{v \in A_u} w_v  \norm{\V}^2  \leq   \frac{1}{\beta} \sum_{v \in V} w_v \norm{\V}^2 \inprod{ \psi(\U), \psi(\V)}
	 =   \frac{1}{\beta} \sum_{v \in V} w_v \norm{\V}^2 \frac{ \inprod{\U,\V}}{ \max \set{ \norm{\V}^2 , \norm{\V}^2 }  } \\
	& \leq   \frac{1}{\beta} \sum_{v \in V} w_v \inprod{\U,\V} \stackrel{\diamond}{=}  \frac{1}{\beta}\leq 1 + \frac{\e}{3}.
\end{align*}
Equality ``$\diamond$'' follows from the SDP constraint $\sum_{v \in V} w_v \inprod{\U,\V} = 1$.
For any $v \in B_u$, we have $\inprod{\psi(\U), \psi(\V)} < \beta$. Hence,
by Property 2 of
orthogonal separators,
$$ \Pr{v \in S_i\Given u \in S_i} \leq \frac{\e}{12k} $$
Therefore,
$$\Ex{ \mu(S_i \cap B_u) \Given u \in S_i } \leq \frac{\e \mu(B_u)}{12k} \leq \frac{\e \mu(V)}{12k} = \frac{\e}{12} .$$
By Markov's inequality,
$\Pr{\mu(S_i \cap B_u) \geq \e /6 \Given u \in S_i } \leq  1/2$.
Since $\mu(S_i)= \mu(S_i\cap A_u) + \mu(S_i\cap B_u)$, we get 
$\Pr{\mu(S_i) \geq 1 + \e/2 \Given u\in S_i}\leq 1/2$.
\end{proof}

\subsection{End of Proof}
\label{sec:merging-sets}
We are ready to finish the analysis of \prettyref{alg:ksets} and prove \prettyref{thm:k-weak-part-alg-1} and \prettyref{thm:bkpart-alg-1}.

\begin{proof}[Proofs of \prettyref{thm:k-weak-part-alg-1} and \prettyref{thm:bkpart-alg-1}]
We first prove \prettyref{thm:k-weak-part-alg-1}, then we slightly modify \prettyref{alg:ksets} and prove 
\prettyref{thm:bkpart-alg-1}.

I. We show that \prettyref{alg:ksets} outputs sets satisfying conditions of \prettyref{thm:k-weak-part-alg-1}.
The sets $S''_i$ are disjoint (see \prettyref{lem:SP-prop}), thus sets $P_i$ are also disjoint. We now need to prove
that among sets $P_i$ obtained at Step~6 of the algorithm, there are at least $(1-\e)k$ 
sets with expansion less than $O_{\e}(\sqrt{\log n \log k}\; OPT)$ (in expectation).

Let $Z=\frac{1}{k}\sum_i \nu(S''_i)$. By \prettyref{lem:SP-prop} we have, 
\[ \Ex{Z}  \leq (8D+1)\, OPT \] 
and $S''_i$ form a partition\footnote{With an exponentially small probability the sets $S''_i$
do not cover all the vertices. In this unlikely event, the algorithm may output an arbitrary partition.}
of $V$. We through away all empty sets $S''_i$, and set $\lambda_i = \mu(S_i'')/k$.
Then $\sum_i \lambda_i = 1$, and
$$Z=\frac{1}{k}\sum_i \nu(S''_i)= \sum_i \lambda_i \cdot \frac{\nu(S''_i)}{\mu(S''_i)}.$$
Define $\cI=\{i: \nu(S''_i)/\mu(S''_i) \leq 3Z/\e\}$. By Markov's inequality (we can think of $\lambda_i$ as the weight of $i$),
\begin{equation}\label{eq:sum-lambda}
\sum_{i\in \cI}\lambda_i \geq 1 -\e/2 \mper 
\end{equation}
Since each $\lambda_i$ satisfies 
\[ \lambda_i = \mu(S''_i)/k \leq (1+\e/2)/k \]
the set $\cI$ has at least $(1-\e/2)k/(1+\e/2)\geq (1 -\e)k$ elements.

Fix an $i\in \cI$. Since $i\in \cI$, we have 
\[ \nu(S''_i)\leq 3Z/\e\cdot \mu(S''_i) \mper \] 
Let $R=\max\{\|\U\|^2: u\in S''_i\}$. 
For a random $r\in (0,R)$ and $L_r = \{u\in S''_i:\|\U\|^2\geq r\}$, we have 
\begin{equation}\label{eq:exp-weight}
 \E_r [w(L_r)]  = \mu(S''_i)/R 
\end{equation}
as each $u$ belongs to $L_r$ with probability $\|\U\|^2/R$ and 
\[ \E_r[|E(L_r, V \setminus L_r)|] = \nu(S''_i)/R \] 
(since an edge in $S''_i\times S''_i$ is cut with probability $|\|\U\|^2 - \|\V\|^2|/R$;
and an edge $(u,v)$ with $u\in S''_i$ and $v\notin S''_i$ is cut with probability $\|\U\|^2$ --- if and only if $u\in L_r$; 
compare with \prettyref{def:nu}). Therefore,
\begin{equation}\label{eq:exp-cut}
\nonumber \E_r[|E(L_r, V \setminus L_r)|]  =  \frac{\nu(S''_i)}{R} \leq \frac{3Z}{\e}\cdot
\frac{\mu(S''_i)}{R} 
  =   \frac{3Z}{\e}\cdot \E_r [w(L_r)] \mper 
\end{equation}
For some $r^*$, we get 
\[ \Abs{ E(L_{r^*}, V \setminus L_{r^*})} \leq  3Z/\e\cdot w(L_{r^*}) \mper \] 
By definition, $\phi_G(P_i) = \min_r \phi_G(L_r)$, thus
\[ \phi_G(P_i)=\frac{\Abs{ E(P_i,V\setminus P_i)}}{w(P_i)}\leq \frac{3Z}{\e} \mper \]

We showed that there are at least $|\cI|\geq (1-\e)k$ sets $P_i$ with expansion at most $3Z/\e$. Therefore, the expansion of the sets 
returned by the algorithm is at most $3Z/\e$. This finishes the proof, since $\Ex{3Z/\e}  = \bigO_{\e} (\sqrt{\log n \log k})\; OPT$.
\\[0.3em]

II. To prove \prettyref{thm:bkpart-alg-1}, we need to modify the algorithm. For simplicity, we 
rescale all weights $w_u$ and assume that $w(G)=k$. Then our goal is to find $k'$ disjoint sets $P_i$ of
weight in the range $[1/2,1+\varepsilon]$ each. Since all sets in the optimal 
solution to the $k$-Balanced Sparsest Partitioning Problem have weight 1, we add the SDP constraint that all vectors $\U$ 
have length 1 (see \prettyref{sec:SDP}): for all $u\in V$:
$$\|\U\|^2 = 1.$$
The intended solution satisfies this constraint. We also change the way the algorithm picks the parameters~$r_i$. The algorithm 
chooses $r_i$ so as to minimize the expansion $\phi_G(P_i)$ subject to an additional constraint $\mu(P_i) \geq (1-\varepsilon/2) \mu(S''_i)$.
Finally, once the algorithm obtains sets $P_i$, it greedily merges sets of weight at most $1/2$.
The rest of the algorithm is the same as \prettyref{alg:ksets}.

From \prettyref{eq:exp-weight} and \prettyref{eq:exp-cut}, we get
\begin{eqnarray*}
\E_r  [w(L_r)] & \geq & \frac{\e^2}{6Z}\nonumber \E_r[|E(L_r, V \setminus L_r)|]+
\frac{(1-\nicefrac{\varepsilon}{2})\mu(S''_i)}{R} \\
& \geq & \max \set{\frac{\e^2}{6Z}\nonumber \E_r[|E(L_r, V \setminus L_r)|],
\frac{(1-\nicefrac{\varepsilon}{2})\mu(S''_i)}{R}}. 
\end{eqnarray*}
Since $\|\U\|^2=1$ for all $u\in V$, we have $R=1$ and $\mu(L_r) = w(L_r)$.
Therefore,
$$\E_r [w(L_r)]
\geq\max\big\{\frac{\e^2}{6Z}\nonumber \E_r[|E(L_r, V \setminus L_r)|],
(1-\frac{\varepsilon}{2})\mu(S''_i)\big\},
$$
and for some $r^*$,
\begin{align*}
w(L_{r^*}) &\geq 
\frac{\e^2}{6Z}\nonumber \E_r[|E(L_{r^*}, V \setminus L_{r^*})|];\\
\mu(L_{r^*}) &\geq (1-\frac{\varepsilon}{2})\mu(S''_i).
\end{align*}
Consequently, we get
$$\phi_G(P_i)\leq \phi_G(L_{r^*}) \leq \frac{6Z}{\e^2}.$$
Now, recall, that by~(\ref{eq:sum-lambda}), $\sum_{i\in\cal I} \lambda_i\geq 1 -\varepsilon/2$. Hence,
\begin{eqnarray*}
\sum_{i\in\cal I} w(P_i) &=& \sum_{i\in\cal I} \mu(P_i) 
\geq (1-\varepsilon/2) \sum_{i\in\cal I} \mu(S''_i)
 =  (1-\varepsilon/2) \sum_{i\in\cal I} k\lambda_i\\
&\geq& (1-\varepsilon) k.
\end{eqnarray*}

We showed that the algorithm gets sets $P_i$ satisfying the following properties: (a) the expansion $\phi_G(P_i)\leq \frac{6Z}{\e^2}$; (b)
$w(P_i)\leq (1+\varepsilon/2)$ and (c) $\sum_i w(P_i)\geq (1- \varepsilon)k$. To get sets of weight in the range $[1/2,1+\varepsilon]$
the algorithm greedily merges sets $P_i$ of weight at most $1/2$ and obtains a collection of new sets, which we denote by $Q_i$.
The algorithm outputs all sets $Q_i$ with weight at least $1/2$.

Note that for any two disjoint sets $A$ and $B$, $\phi_G(A\cup B)\leq \max\{\phi_G(A), \phi_G(B)\}$. 
So $\phi_G(Q_i)\leq \max_j \phi_G(P_j)\leq \frac{6Z}{\e^2}$. All sets $Q_i$ but possibly one have weight at least $1/2$. So the weight of sets $Q_i$ output by the algorithm is at least $(1-\varepsilon)k -1/2$. The maximum weight
of sets $Q_i$ is $1+\varepsilon/2$, so the number of sets $Q_i$ is at least
$$\ceil{\frac{(1-\varepsilon)k -1/2}{1-\varepsilon/2}}
\geq \ceil{(1-2\varepsilon)k -1/2}\geq \ceil{(1-4\varepsilon)k}.$$
To verify the last inequality check two cases: if $2\varepsilon k\geq 1/2$, then $(1-2\varepsilon)k -1/2 \geq (1-4\varepsilon)k$; if $2\varepsilon k < 1/2$, then $\ceil{(1-2\varepsilon)k -1/2} = k$.
This finishes the proof.
\end{proof}

\section{From Disjoint Sets to Partitioning}\label{sec:partit}
We now show how given $k'\geq (1-\e)$ sets $P_1,\dots,P_{k'}$, we  can obtain
a true partitioning $P'_1,\dots,P'_{k''}$ of $V$.

\begin{proof}[Proof of \prettyref{thm:kpart-alg-1}]
To get the desired partitioning, we first run \prettyref{alg:ksets} several times (say, $n$) to obtain 
disjoint non-empty sets $P_1,\dots, P_{k'}$ that satisfy $\max_i \phi_G(P_i) \leq \bigO_{\e} (\sqrt{\log n \log k})\; OPT$
w.h.p. Let $Z = \max_i \phi_G(P_i)$. We sort sets $P_i$ by weight $w(P_i)$. We output the smallest 
$k'' = \floor{(1-\e) k'}$ sets $P_i$, and the compliment set $P' = V\setminus (\cup_{1\leq i\leq k''} P_i)$.

Since sets $P_i$ are disjoint and non-empty, the first $k''$ sets $P_i$ and the set $P'$ are also disjoint and non-empty. Moreover,
$\phi_G(P_i)\leq Z$, so we only need to show that $\phi_G(P')\leq \bigO_{\e}(Z)$. Note, that $w(P')\geq \e w(V)$, since $P'$ 
contains vertices in the $\ceil{\e k}$ largest sets $P_i$ and all vertices not covered by sets $P_i$.
Then, 
\[ E(P', V\setminus P') = \cup_{i\leq k''} E(P', P_i) \subset \cup_{i\leq k''} E(P_i, V\setminus P_i) \mper \] So
\begin{align*}
\phi_G(P')&=\frac{|E(P', V\setminus P')|}{w(P')} 
\leq \frac{\sum_{i=1}^{k''} E(P_i, V\setminus P_i)}{w(P')} \\
&= \frac{\sum_{i=1}^{k''} w(P_i)\phi_G(P_i)}{\e w(V)}
\leq 
\frac{\sum_{i=1}^{k''} w(P_i)Z}{\e w(V)} \\
& \leq \frac{Z w(V)}{\e w(V)} = \frac{Z}{\e}.
\end{align*}
This concludes the proof.
\end{proof}

\section{Acknowledgements}
The first author would like to thank Prasad Raghavendra, Piyush Srivastava and Santosh Vempala for helpful discussions.

\bibliography{smallset}
\bibliographystyle{amsalpha}
\appendix
\section{Proof of \prettyref{thm:kpart-alg-2}}
\label{app:kpart-alg-2}
The proof of \prettyref{thm:kpart-alg-2} is almost the same as the proof of \prettyref{thm:kpart-alg-1}. The only difference is that
we need to replace orthogonal separators with a slightly different variant of orthogonal separators (implicitly defined in~\cite{cmm06b}).

{\em {\bf Orthogonal Separators with $\ell_2$ distortion.}}
Let $X$ be a set of unit vectors in $\ell_2$. We say that a distribution over subsets of $X$ is a $k$-orthogonal
separator of $X$ with $\ell_2$ distortion $D$, probability scale $\alpha >0$ and separation threshold $\beta<1$, if the following
conditions hold for $S \subset X$ chosen according to this distribution:
\begin{enumerate}
\item For all $\U \in X$, $\Pr{ \U \in S} = \alpha$.

\item For all $\U, \V \in X$ with $\inprod{\U,\V} \leq \beta \max \set{ \norm{\U}^2, \norm{\V}^2 }$,

$$\Pr{\U \in S \textrm{ and } \V \in S} \leq \frac{\alpha}{k} .$$

\item For all $u,v \in X$,
$$\Pr{ I_S(\U) \neq I_S(\V) } \leq \alpha D \|\U - \V\|. $$
\end{enumerate}

\begin{theorem}[\cite{cmm06b}]\label{thm:ort-sep-prime}
There exists a polynomial-time randomized algorithm that given a set
of unit vectors $X$, a parameter $k$, and $\beta < 1$ generates a $k$-orthogonal separator with $\ell_2$ distortion
$D = \bigO_{\beta} \left( \sqrt{\log k}   \right)$
and scale $\alpha \geq 1/n$.
\end{theorem}

For completeness we sketch the proof of this lemma in \prettyref{sec:ort-sep-prime}. \prettyref{alg:ksets}$'$ is the same
as \prettyref{alg:ksets} except that at Step 3, it samples orthogonal separators with 
$\ell_2$ distortion $\bigO_{\e}(\sqrt{\log k})$ using \prettyref{thm:ort-sep-prime}. 
The proof of \prettyref{thm:kpart-alg-1} goes through for the new algorithm essentially as is. The only 
statement we need to take care of is \prettyref{lem:SP-prop}~(a). We prove the following
bound on $\Ex{\sum_i \nu(S''_i)}$.

\begin{lemma}
The sets $S_i''$ satisfy the following condition:
$\Ex{\sum_i \nu(S''_i)} \leq (8D +1)k \cdot \sqrt{SDP}$, where 
$D=\bigO_{\e}(\sqrt{\log k})$ is the $\ell_2$ distortion of $(12k/\e)$-orthogonal separator, and $SDP$ is the value
of the SDP solution.
\end{lemma}
\begin{proof}
Let $E_{cut}$ be the set of edges cut by the partitioning $S_1'',\dots, S_T'', V\setminus (\cup S_i'')$. 
As before (in \prettyref{lem:SP-prop}), we have
\begin{eqnarray*}
\Ex{\sum_i \nu(S_i'')} 
& \leq &  \Ex{\sum_{(u,v)\in E_{cut}} (\|\U\|^2 + \|\V\|^2)} + \sum_{(u,v)\in E} |\|\U\|^2 - \|\V\|^2|  \\
&\leq & \Ex{\sum_{(u,v)\in E_{cut}} (\|\U\|^2 + \|\V\|^2)} + k\, SDP \mper 
\end{eqnarray*}

We now bound the first term. 
Estimate the probability that an edge $(u,v)$ is cut. Let $U_t = \cup_{i\leq t}S'_i$ be the set of vertices covered by the first $t$ sets $S'_i$. Note, that $S''_i = S'_i\setminus U_{i-1}$.
We say that the edge $(u,v)$ is cut by the set $S'_t$, if $S'_t$ is the first set containing $u$ or $v$, and it contains only one of these vertices.
Then,
\begin{eqnarray*}
 \Pr {(u,v)\in E_{cut}} & = &  \sum_i \Pr {(u,v)\text{ is cut by }S'_i\,} \\ 
& = &  \sum_{i} \Pr {u,v \notin U_{i-1} \text{ and } I_{S'_i}(u) \neq I_{S'_i}(v)} \\
&\leq &  \sum_{i} \Pr {u \notin U_{i-1} \text{ and } I_{S_i}(u) \neq I_{S_i}(v)} \\
& = & \sum_{i} \Pr {u \notin U_{i-1}} \Pr{I_{S_i}(u) \neq I_{S_i}(v)}. 
\end{eqnarray*}
Now, by \prettyref{lem:min-prob}, $\Pr {u \notin U_{i-1}}\leq (1-\alpha/2)^{i-1}$, and, by Property 3 of $\ell_2$ orthogonal separators,
\[ \Pr{I_{S_i}(u) \neq I_{S_i}(v)} \leq  \alpha D \norm{\psi(u)-\psi(v)} 
\leq  \alpha D \frac{\sqrt{2}\,\norm{\U - \V} }{\max \set{ \norm{\U}, \norm{\V} }} \mper \]
Thus,
$$\Pr {(u,v)\in E_{cut}} \leq \frac{2\sqrt{2}\,D\,\|\U-\V\|}{\max\{\|\U\|, \|\V\|\}}.$$
Now, the proof deviates from the proof of \prettyref{lem:SP-prop}:
\begin{eqnarray*}
\Ex{\sum_{(u,v)\in E_{cut}} (\|\U\|^2 + \|\V\|^2)} \
& = & \sum_{(u,v)\in E} \Pr {(u,v)\in E_{cut}} (\|\U\|^2 + \|\V\|^2) \\
& \leq & \sum_{(u,v)\in E} \frac{2\sqrt{2}\,D\,\|\U-\V\|}{\max\{\|\U\|, \|\V\|\}} \cdot 
(\|\U\|^2 + \|\V\|^2)  \\
&\leq & 2\sqrt{2}\,D\, \sum_{(u,v)\in E} \|\U-\V\|\cdot(\|\U\| + \|\V\|).  
\end{eqnarray*}
By Cauchy--Schwarz,
\begin{eqnarray*}
 2\sqrt{2}\,D\, \sum_{(u,v)\in E} \|\U-\V\|\cdot(\|\U\| + \|\V\|) 
& \leq &  2\sqrt{2}\,D\, \Big(\sum_{(u,v)\in E} \|\U-\V\|^2\Big)^{{1/2}}  
\Big(\sum_{(u,v)\in E} (\|\U\| + \|\V\|)^2\Big)^{{1/2}}  \\
& \leq & 4\,D\, \Big(\sum_{(u,v)\in E} \|\U-\V\|^2\Big)^{{1/2}}  
\Big(\sum_{(u,v)\in E} \|\U\|^2 + \|\V\|^2\Big)^{{1/2}}  \\
&= & 4\,D\, \Big(k\, SDP\Big)^{1/2} 
\Big(\sum_{(u,v)\in E} d_u \|\U\|^2\Big)^{1/2}. 
\end{eqnarray*}

Recall, that in \prettyref{thm:kpart-alg-2}, we assume that the weight of every vertex $w_u$ equals its degree $d_u$. Hence,
$\sum_{(u,v)\in E} d_u \|\U\|^2 = \sqrt{\mu(V)} = \sqrt{k}$. We get,
\begin{eqnarray*}
\Ex{\sum_{(u,v)\in E_{cut}} (\|\U\|^2 + \|\V\|^2)} &\leq& 
4\,D\, k\sqrt{SDP}.
\end{eqnarray*}
Since $SDP\leq OPT = \phi^k_G \leq 1$ (here we use that $d_u=w_u$), $SDP\leq \sqrt{SDP}$, and
\[
\Ex{\sum_i \nu(S''_i)}  \leq   8Dk\,\sqrt{SDP} + k\, SDP 
  \leq  (8D+1)k\,\sqrt{SDP}. \]
This concludes the proof.
\end{proof}


\section{\texorpdfstring{Orthogonal Separators with $\ell_2$ Distortion}{Orthogonal Separators with L2 Distortion}}\label{sec:ort-sep-prime}
In this section, we sketch the proof of \prettyref{thm:ort-sep-prime} which is proven in~\cite{cmm06b}
as part of Lemma~4.9. Let us fix some notation. Let $\bar \Phi(t)$ be the probability that 
the standard $\cN(0,1)$ Gaussian variable is greater than $t$. We 
will use the following easy lemma from~\cite{mm12}.

\begin{lemma}[Lemma 2.1. in ~\cite{mm12}]
For every $t>0$ and $\beta \in (0,1]$, we have
$$\bar\Phi(\beta t)\leq \bar\Phi(t)^{\beta^2}.$$
\end{lemma}

We now describe an algorithm for $m$-orthogonal separators with $\ell_2$ distortion (see \prettyref{app:kpart-alg-2}).
Let $\beta < 1$ be the separation threshold. Assume w.l.o.g. that all vectors $\U$ lie in~$\mathbb{R}^n$.
Fix $m' = m^\frac{1+\beta}{1-\beta}$ and $t =\bar\Phi^{-1}(1/m')$ (i.e., $t$ such that $\bar\Phi(t) = 1/m'$). Sample a random Gaussian $n$ dimensional vector $\gamma$ in $\mathbb{R}^n$. Return the set
$$S = \{\U : \inprod{ \U, \gamma} \geq t\}.$$

We claim that $S$ is an $m$-orthogonal separator with $\ell_2$ distortion $\bigO(\sqrt{\log m})$ and scale 
$\alpha = 1/m'$. 
We now verify the conditions of orthogonal separators with $\ell_2$ distortion.

1. For every $\U$, 
$$\Pr{\U \in S} = \Pr{\inprod{ \U, \gamma } \geq t} = 1/m'\equiv \alpha.$$
Here we used that $\inprod{ \U, \gamma }$ is distributed as $\cN(0,1)$, since $\U$ is a unit vector.

2. For every $\U$ and $\V$ with $\inprod{ \U,\V } \leq \beta$,
\begin{eqnarray*}
\Pr{\U, \V \in S} & =&  \Pr{\inprod{\U, \gamma} \geq t \text{ and } \inprod{\V,\gamma} \geq t} \\ 
 & \leq &  \Pr{\inprod{ \U +\V, \gamma} \geq 2t}.
\end{eqnarray*}
Note that $\|\U +\V\|=\sqrt{2+2\inprod{ \U, \V}}$, hence 
$(\U +\V)/\sqrt{2+2\inprod{ \U, \V}}$ is a unit vector. We have
\begin{eqnarray*}
 \Pr{\U, \V \in S} 
& \leq & \ProbOp\Big[\inprod{ \frac{\U +\V}{\sqrt{2+2\inprod{\U,\V} }}, \gamma } \geq \frac{2t}{\sqrt{2+2\inprod{\U,\V} }}\Big] \\
& = & \bar\Phi \Big(\frac{\sqrt{2}t}{\sqrt{1+\inprod{ \U, \V}}}\Big) \leq \bar\Phi \Big(\frac{\sqrt{2}t}{\sqrt{1+\beta}}\Big) \leq 
\bar\Phi(t)^{\frac{2}{1+\beta}} \\
& = & \Big(\frac{1}{m'}\Big)^{\frac{2}{1+\beta}} 
= \frac{1}{m'}\cdot \Big(\frac{1}{m'}\Big)^{\frac{1-\beta}{1+\beta}} = \frac{\alpha}{m}. 
\end{eqnarray*}

3. The third property directly follows from Lemma A.2. in~\cite{cmm06b}.

\medskip

We note that this proof gives probability scale $\alpha = m^{-\frac{1+\beta}{1-\beta}}$. So, for some $\beta$,
we may get $\alpha \ll 1/n$. However, it is easy to sample $\gamma$ in such a way that 
$\Pr{\inprod{ \U, \gamma} \geq 1/n}$ for every vector $\U$ in our set. To do so, 
we order vectors $\{\U\}$ in an arbitrary way: $\U_1,\dots, \U_n$. Then, we pick a random index 
$\iota\in\{1,\dots, n\}$, and sample a random Gaussian vector $\gamma'$ conditional on $\inprod{ \U_{\iota}, \gamma' } \geq t$. We set $S' = \{\U : \inprod{ \U, \gamma'} \geq t\}$ as in the algorithm above.
Note that $\U_\iota$ always belongs to $S'$.
We output $S''=S'$ if $S'$ does not contain vectors $\U_1,\dots, \U_{\iota-1}$; and we output
$S''=\varnothing$ otherwise. It is easy to verify that $\Pr{\U \in S''}=1/n$ for every $\U$, and, furthermore,
for every non-empty set $S^*\neq \varnothing$,
$$\Pr{S'' = S^*}= \frac{1}{\alpha n}\Pr{S = S^*},$$
where $S$ is the orthogonal separator from the proof above. So all properties of orthogonal
separators hold for $S''$ with $\alpha' = \alpha/(\alpha n) = 1/n$.

\section{Proof of \prettyref{prop:cheeger-partition}}
\label{app:cheeger-partition}

We restate \prettyref{prop:cheeger-partition} below. 

\begin{proposition}
Given a graph $G$ and a parameter $k$,
\[ \lambda_k \leq \phik \leq \bigO_{\e}\left(\sqrt{ \lambda_{(1+\e)k} \log k }\right) \mper \] 
for every $\e>0$. 
\end{proposition}

\begin{proof}
\cite{lot12} show that there exist disjoint non-empty sets $P_1,\dots, P_{k'}$ that satisfy 
$\max_i \phi_G(P_i) \leq \bigO_{\e} (\sqrt{\lambda_{(1+\e) k} \log k})$ for $k' \geq k(1 + \e/2)$. 
Let $Z = \max_i \phi_G(P_i)$. We sort sets $P_i$ by weight $w(P_i)$. We output the smallest 
$k$ sets $P_i$, and the compliment set $P' = V\setminus (\cup_{1\leq i\leq k} P_i)$.

Since sets $P_i$ are disjoint and non-empty, the first $k$ sets $P_i$ and the set $P'$ are also disjoint and non-empty. Moreover,
$\phi_G(P_i)\leq Z$, so we only need to show that $\phi_G(P')\leq \bigO_{\e}(Z)$. Note, that $w(P')\geq \e w(V)$, since $P'$ 
contains vertices in the $\ceil{\e k}$ largest sets $P_i$ and all vertices not covered by sets $P_i$.
Then, 
\[ E(P', V\setminus P') = \cup_{i\leq k} E(P', P_i) \subset \cup_{i\leq k} E(P_i, V\setminus P_i) \mper \] So
\begin{align*}
\phi_G(P')&=\frac{|E(P', V\setminus P')|}{w(P')} 
\leq \frac{\sum_{i=1}^{k} E(P_i, V\setminus P_i)}{w(P')} \\
&= \frac{\sum_{i=1}^{k} w(P_i)\phi_G(P_i)}{\e w(V)}
\leq 
\frac{\sum_{i=1}^{k} w(P_i)Z}{\e w(V)} \\
& \leq \frac{Z w(V)}{\e w(V)} = \frac{Z}{\e}
.
\end{align*}
This concludes the proof.

\end{proof}

\section{Integrality Gap for the Assignment SDP}
\label{app:ug-ig}

In this Section, we show that the  standard {\em Assignment SDP} has high integrality gap.

\begin{figure}[ht]
\begin{tabularx}{\columnwidth}{|X|}
\hline
\[ \min \alpha \]
\begin{eqnarray*}
\sum_{(u,v) \in E} \norm{\U_i - \V_i}^2 & \leq & \alpha \sum_{u \in V} w_u \norm{\U_i}^2 \qquad \forall i \in [k] \\
\sum_{i \in [k]} \norm{\U_i}^2 & = & 1 \\
\inprod{\U_i,\U_j} & = & 0 \qquad \forall i \neq j \textrm{ and } \forall u \in V \\
\inprod{\sum_i \U_i,I} & = & 1 \\
\norm{I}^2 & = & 1\\
\end{eqnarray*}
\\
\hline 
\end{tabularx}
\caption{Assignment SDP}
\label{sdp:assignment}
\end{figure}

\begin{proposition}
\prettyref{sdp:assignment} has an unbounded integrality gap. 
\end{proposition}

\begin{proof}
Consider the following infinite family of graphs $\cG = \set{G_n : n \geq 0}$.
$G_n$ consists of the two disjoint cliques of size $C_1 = K_{\floor{n/2}}$ and $C_2 = K_{\ceil{n/2}}$. 
It is easy to see that for $\phi^k(G_n) = \Omega(1)$ for $k > 2$.

For the sake of simplicity, let us assume that $k$ is a multiple of $2$. 
Let $e_1, \ldots, e_{k/2}$ be the standard basis vectors. 
Consider the following vector solution to \prettyref{sdp:assignment}.
\[  \U_i = \begin{cases}   
		\sqrt{\frac{2}{k}}  e_i   & \textrm{ if } u \in C_1 \textrm{ and } i \leq k/2 \\  
		\sqrt{\frac{2}{k}}  e_{(i-k/2)}   & \textrm{ if } u \in C_2 \textrm{ and } i > k/2 \\
		0	& \textrm{ otherwise} \\
 \end{cases}  \]
and
\[ I = \sqrt{\frac{2}{k}} \sum_{i=1}^{k/2} e_i \mper \]
It is easy to verify that this is a feasible solution with $\alpha = 0$. 
Therefore, \prettyref{sdp:assignment} has an unbounded integrality gap.

\end{proof}

\end{document}